\documentclass[letterpaper, 10 pt, twocolumn]{IEEEtran}  

\IEEEoverridecommandlockouts                              

\usepackage[colorlinks,urlcolor=blue,linkcolor=blue,citecolor=blue]{hyperref}
\usepackage{color,array}
\usepackage{graphicx}
\usepackage{float} 

\usepackage[utf8]{inputenc}
\usepackage[T1]{fontenc}
\usepackage{soul}

\usepackage{amsthm}
\usepackage{amsmath}
\usepackage{amsfonts}
\usepackage{amssymb}

\usepackage{kbordermatrix}
\hypersetup{colorlinks=true, linkcolor=blue, filecolor=magenta, urlcolor=cyan,}
\usepackage{mathrsfs}
\usepackage{bbold}

\usepackage[export]{adjustbox}
\usepackage{todonotes}

\usepackage{algorithm}
\usepackage{algpseudocode}

\usepackage{subcaption}

\usepackage{tikz}
\usetikzlibrary{automata,positioning}
\usetikzlibrary{shapes,arrows}
\tikzset{state/.style={circle, draw, minimum size=0.5cm, initial distance=0.2cm, initial text=},}
\usepackage{dirtytalk}
\usepackage{xcolor}
\usepackage{alphabeta}
\usepackage[]{mdframed}
\usepackage{cite}
\usepackage{multirow}
\usepackage{array}

\usepackage{xparse}    
\newtheorem{definition}{Definition}

\newtheorem*{informaltheorem}{Theorem 1 (Informal)}
\newtheorem{formaltheorem}{Theorem (Formal)}

\newtheorem{example}{Example}
\newtheorem{lemma}{Lemma}

\newtheorem{corollary}{Corollary}

\newcommand{\algmargin}{\the\ALG@thistlm}

\algnewcommand{\LineComment}[1]{\Statex \(\triangleright\) #1}
\algnewcommand{\parState}[1]{\State%
	\parbox[t]{\dimexpr\linewidth-\algmargin}{\strut #1\strut}}
\algnewcommand{\parRequire}[1]{\Require%
	\parbox[t]{\dimexpr\linewidth-\algmargin}{\strut #1\strut}}
\usepackage{enumitem}
\setlist{nosep,leftmargin=*} 

\usepackage[font=footnotesize,labelfont=bf]{caption}
\usepackage{enumitem}
\setlist{nosep,leftmargin=*}                        

\makeatletter
\g@addto@macro\normalsize{%
  \setlength{\abovedisplayskip}{4pt}%
  \setlength{\belowdisplayskip}{6pt}%
  \setlength{\abovedisplayshortskip}{4pt}%
  \setlength{\belowdisplayshortskip}{6pt}%
}
\makeatother

\newif\iflongversion
\longversiontrue         

\iflongversion\else
  \RenewDocumentEnvironment{proof}{ o +b }{}{}
\fi

\newlength{\whilewidth}
\algdef{SE}[parWHILE]{parWhile}{EndparWhile}[1]
{\parbox[t]{\dimexpr\linewidth-\algmargin}{%
		\hangindent\whilewidth\strut\algorithmicwhile\ #1\ \algorithmicdo\strut}}{\algorithmicend\ \algorithmicwhile}
\title{\LARGE \bf
Data-driven Supervisory Control under Attacks via Spectral Learning
}

\author{Nathan Smith and Yu Wang$^{1}$
\thanks{$^{1}$Nathan Smith and Yu Wang are with the Department of Mechanical \& Aerospace Engineering, University of Florida, Gainesville, FL, USA. Email: {\tt\small \{nsmith3, yuwang1\}@ufl.edu}}}

\renewcommand{\kbldelim}{[}
\renewcommand{\kbrdelim}{]}

\begin{document}

\maketitle
\thispagestyle{empty}
\pagestyle{plain}

\begin{abstract}

The technological advancements facilitating the rapid development of cyber-physical systems (CPS) also render such systems vulnerable to cyber attacks with devastating effects. Supervisory control is a commonly used control method to neutralize attacks on CPS. The supervisor strives to confine the (symbolic) paths of the system to a desired language via sensors and actuators in a closed control loop, even when attackers can manipulate the symbols received by the sensors and actuators. Currently, supervisory control methods face limitations when effectively identifying and mitigating unknown, broad-spectrum attackers. In order to capture the behavior of broad-spectrum attacks on both sensing and actuation channels we model the plant, supervisors, and attackers with finite-state transducers (FSTs). Our general method for addressing unknown attackers involves constructing FST models of the attackers from spectral analysis of their input and output symbol sequences recorded from a history of attack behaviors observed in a supervisory control loop. To construct these FST models, we devise a novel learning method based on the recorded history of attack behaviors. A supervisor is synthesized using such models to neutralize the attacks.

\end{abstract}

\vspace{-12pt}
\section{Introduction}
\label{sec:intro}

Supervisory control is a commonly used control method for cyber-physical systems (CPS) which consists of the combination of digital and physical processes such as smart grids, autonomous vehicles, and precision agriculture \cite{wang_state-based_2018, zhang_perception_2023, diazredondo2020security, jazdi2014cyber}. Traditional continuous control techniques are insufficient for managing (CPS) due to their discrete event-driven nature. Unlike continuous systems, CPS are characterized by discrete events such as keyboard inputs and sensor readings, which trigger state changes \cite{cassandras2008introduction}. In CPS, adherence to specific logical rules expressed using regular expressions is often prioritized over following predefined trajectories, a task more aligned with classical control methodologies. To address this, supervisors function as high-level controllers, facilitating communication between sensors and actuators by exchanging discrete symbols \cite{wonhamSupervisoryControlDiscreteevent2018}. Their role is to enforce abstracted logical rules represented by regular expressions.


CPS are often subject to attacks that can have devastating effects; hence, supervisory control must be studied within the context of attacks. For example, the Stuxnet Worm severely damaged an Iranian power plant in 2010 \cite{langner2011stuxnet}, and a CPS attack led to the shutdown of a Ukrainian power grid in 2017 ~\cite{whitehead2017ukraine, soltan2019react}. In the United States alone CPS attacks are estimated to have caused 6.9 billion dollars worth of damages in 2021 compared to only 32 million in 2010. Similarly, the CPS recorded attack frequency has risen from 1.4 per year to 5.8 per year from 2010 to 2021 \cite{sharif2022literature}. Such CPS attacks are represented with supervisory control by the real-time manipulation of sensor and actuation symbols during communication, e.g. man-in-the-middle and network attacks \cite{alguliyevCyberphysicalSystemsTheir2018, singhStudyCyberAttacks2018}. Traditionally, supervisory control under attacks involves either supervisor synthesis (defensive goal) or attacker synthesis (offensive goal). Supervisor synthesis focuses on \say{creating a supervisor that restricts the plant to desired behaviors despite the presence of specific attacker classes} \cite{you2022livenessenforcing, carvalho2018detection, wang2023attackresilient}. Conversely, attacker synthesis tackles an adversarial problem, aiming to \say{design an attacker that forces the plant into undesired behavior under supervision}  \cite{taiSynthesisOptimalCovert2023, meira-goesSynthesisSensorDeception2020}. Regarding attack synthesis, the design of an attacker based on control loop observations is considered in \cite{tai2023synthesis}, and attackers that can remain stealthy are studied in \cite{taiPrivacypreservingCosynthesisSensor2023, lin2020synthesis, lin_synthesis_2021, meira-goes2021synthesis}.

Currently, supervisory control methods face limitations in effectively identifying and mitigating unknown, broad-spectrum attackers; this motivated us to focus on the more rigorous supervisor synthesis problem rather than the attacker synthesis problem \cite{duo2022survey}. In essence, most supervisory design approaches are highly specialized to address particular classes of attacks or specific control objectives \cite{carvalho2018detection, wakaikiSupervisoryControlDiscreteEvent2019, rashidinejadSupervisoryControlDiscreteEvent2019, su2018supervisor, limaSecurityCyberPhysicalSystems2022, taiPrivacypreservingCosynthesisSensor2023}. To our knowledge, there are few developments outside this work to design supervisors under general attacks with both unknown sensor and actuator attacks. For supervisor design under attacks, most previous work \cite{carvalho2018detection, wakaikiSupervisoryControlDiscreteEvent2019, rashidinejadSupervisoryControlDiscreteEvent2019, su2018supervisor, limaSecurityCyberPhysicalSystems2022} require knowing the attacker model. The work in ~\cite{thapliyal2021learning} discusses cyberattack and defense design for supervisory controllers in CPSs when partial information is available to both the attacker and the supervisor. They also study supervisor design when the attack model is unknown, but their work has key differences from our ideas. In their setup, the plant is modeled as an LTI system and the supervisor as an automaton responsible for producing the control inputs for the plant. Although the supervisors in \cite{su2018supervisor} can handle unknown attacks, they are confined to a class of bounded attacks. In addition, the supervisors in \cite{su2018supervisor} can only handle attacks on the sensor symbols. The design of supervisors that maintain liveness properties under attack are explored in \cite{you2022livenessenforcing}. In \cite{yaoSensorDeceptionAttacks2022, taiPrivacypreservingCosynthesisSensor2023}, the supervisor is designed to uphold plant privacy under attack. The developments made in this work can be applied to a wide variety of supervision goals that encompass both privacy and liveness.

We focus on a model-based approach where finite state transducers (FSTs) are used to model the unknown attacker \cite{mohri_weighted_2004}.  Many common CPS attacks are governed by a finite set of temporal rules, that are well represented by finite state machines such as FSTs \cite{diazredondo2020security, battistelli2009unfalsified}. In order to capture this behavior of broad-spectrum attacks, we model attackers with FSTs.  FSTs are related to but generalize the commonly-used finite-state automata in supervisory control \cite{cassandras2008introduction} and are discussed in Section \ref{sec:prelim}. Specifically, FST can be abstracted as a graph whose transitions between states are driven by an input symbol and produce a different output. This versatility enables FSTs to encapsulate various types of attack, ranging from basic incursions such as symbol insertion, deletion, or substitution to more complex maneuvers such as replay attacks \cite{wang2023attackresilient}. Moreover, FSTs are adept at modeling dynamic assaults like the Stuxnet worm \cite{langner2011stuxnet}, and those characterized by man-in-the-middle or network-based strategies. The ability of FSTs to represent bi-directional communication also enables them to model the plant, attackers and supervisor in a closed control loop. This versatile structure enables the application of the same supervisor synthesis strategy across a wide range of FSTs \cite{wang2023attackresilient}. Expanding on the supervisor design frameworks developed in \cite{wang2019supervisory,wang2019attackresilient,wang2023attackresilient}, we enhance their adaptability to scenarios where attacker FST models are unknown.

We investigate the following question: how can a set of sample attacks be used to design a resilient supervisor? We say a supervisor is resilient if it is able to enforce the desired behavior on the plant regardless of the attack mode produced by an attacker FSTs. We consider attacks on both the sensor and actuation communication channels of the supervisory control loop shown in \ref{fig:control_loop_record}. CPS go through extensive testing phases, allowing for the collection of sample attacks. Attacks can occur during software-in-the-loop testing with a real adversary \cite{raghupatruni_empirical_2019, he_system_2019} or simulated during penetration testing \cite{li_deep_2022}. Software in-the-loop testing has shown significant utility for the automotive production industry. Penetration testing exposes the CPS to a diverse set of cyber attacks to find what may be the most hazardous in a controlled environment. This allows the attack behavior to be learned in a low-risk environment preventing harmful effects from the attacker. These samples could then be used to learn a FSM representation of the attacker in an offline environment.


\begin{figure}[b]
\centering
\includegraphics[width=0.9\linewidth]{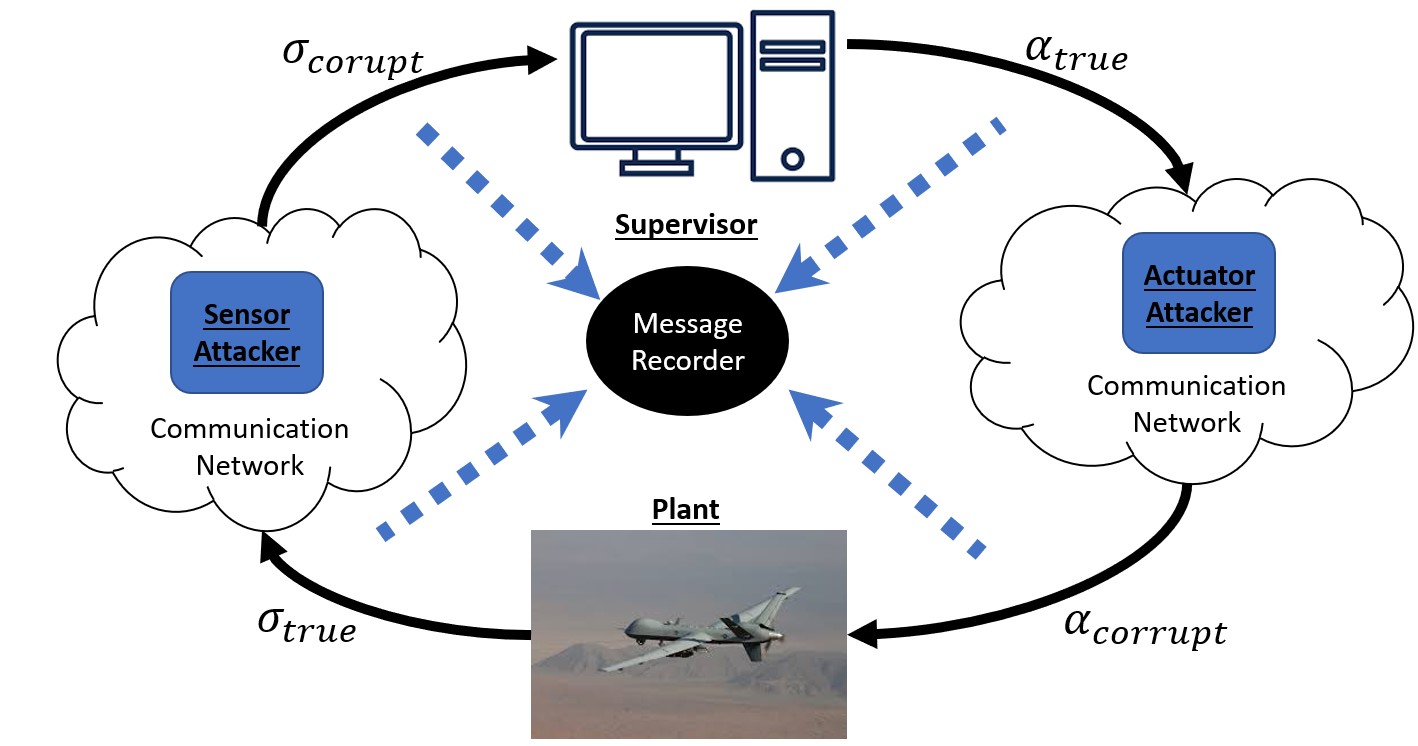}
\caption{\small Supervisory control of cyber-physical systems under attacks. Attackers intercept messages between Plant and Supervisor. Attacks recorded during testing phases.} 
\label{fig:control_loop_record}
\end{figure}


The novelty of this manuscript is a data-driven method that leverages spectral learning to ultimately provide a supervisor resilient to unknown sensor and actuator attackers. Our general technical method for handling unknown attackers is based on constructing FST models of the attackers by performing spectral analysis on collected sequences of their input and output symbols. Traditionally, the learning of a FSM from samples is done through variations of the $L^*$ algorithm by relying on queries\cite{angluinLearningRegularSets1987}. However, we find it unreasonable to assume an attacker would cooperate with such queries. Instead, our method leverages spectral learning, which does not depend on queries and has demonstrated effectiveness in learning matrix representations of FSMs directly from data. \cite{hsu_spectral_2012, balle_spectral_2014}. Our approach begins by organizing the collected sample attacks into binary matrices that represent the attacker behavior in Section \ref{sec:data_requirements}. Since we consider attackers modeled with FSTs the attack behavior patterns are described by a finite graph structure. This finite graph structure can be represented with 2D arrays known as Hankel matrices, where each linearly independent row is representative of a state in the attacker FST \cite{balle_learning_2015}. Spectral learning is then applied to such matrices to extract their linearly independent rows and eventually obtain the FST model in Section \ref{sec:spec_method}. This model is then used to design a supervisor that is resilient to attacks captured by the model. In addition to the attacker model the existence of such a supervisor is dependent on the dynamics and desired behavior of the plant.

\vspace{-12pt}
\section{Preliminaries}
\label{sec:prelim}

We introduce finite-state transducers (FSTs), which are used to model broad-spectrum attackers. They are also used to model the plant and supervisor to provide a more general supervisory control loop than the classic automata-based supervisory control loop. FST's are carefully defined in \ref{subsec:prelim_FST} and their use within the supervisory control loop is presented in \ref{subsec:prelim_SCL}.
\vspace{-12pt}
\subsection{Finite state transducers}
\label{subsec:prelim_FST}
\begin{definition}
\label{def:fst}
A finite-state transducer (FST) is a tuple $(\mathbf{S}, \mathbf{s_0}, \mathbf{I}, \mathbf{O}, \mathbf{Trans}, \allowbreak \mathbf{S_f})$ where
\begin{itemize}
  \item $\mathbf{S}$ is a finite set of states;
  \item $s_0 \in \mathbf{S}$ is the initial state;
  \item $\mathbf{I}$ is a finite set of non-empty input symbols that contains the empty symbol $\varepsilon$;
  \item $\mathbf{O}$ is a finite set of non-empty output symbols that contains $\varepsilon$;
  \item $\mathbf{Trans} \subseteq \mathbf{S} \times \mathbf{I} \times\mathbf{O} \times \mathbf{S}$ is a transition relation;
  \item $(s, \epsilon, \epsilon, s) \in \mathbf{Trans}$ for all $s \in \mathbf{S}$;
  \item $\mathbf{S_f} \subseteq \mathbf{S}$ is a set of final states.
\end{itemize}
\end{definition}  

We call $\mathbf{I} \times \mathbf{O}$ the \emph{alphabet} of the FST. The Kleene Star of an alphabet $\mathbf{I} \times \mathbf{O}$ is defined as $(\mathbf{I} \times \mathbf{O})^* = \{\chi_1 \ldots \chi_n \mid \chi_1, \ldots, \chi_n \in \mathbf{I} \times \mathbf{O}, n \in \mathbb{N} \}$; by convention, $(\varepsilon, \varepsilon) \in (\mathbf{I} \times \mathbf{O})^*$. A \emph{path} of an FST is defined as a sequence of transitions $\{(s_{k}, i_k, o_k, s_{k+1} )\}^n_{k = 1}$ where $(s_{k}, i_k, o_k, s_{k+1}) \in \mathbf{Trans}$ for all $1 \leq k \leq n$ and $s_{1} = \mathbf{s_0}$. 
An FST path is \emph{accepted} if $s_{{n+1}} \in \mathbf{S_f}$
This accepted path defines an accepted \emph{word} of the FST with the concatenation of input/output pairs, $\chi_1 \chi_2 ... \chi_n$ where $\chi_k = (i_k, o_k)$ for $k = 1, \ldots, n$. The language of an FST is the set of all words accepted by the FST, denoted by $L( \mathcal{A} )$. The FST language describes what combination of input and outputs are accepted by the FST and is used to differentiate one FST from another. Two FSTs, $\mathcal{A}$ and $\mathcal{A'}$, are equivalent if $L(\mathcal{A}) = L(\mathcal{A'})$.  The minimum state equivalent of an FST $\mathcal{A}$ is an FST $\mathcal{A'}$ equivalent to $\mathcal{A}$ with the fewest possible states. Unless explicitly stated otherwise we will only consider the minimum state equivalent of FSTs.

We also call $i_1 i_2 \dots i_n$ an input word and $o_1 o_2 \dots o_n$ an output word. The set of allowed input words is the \emph{input language} and denoted by $L_{in}(\mathcal{A})$. Similarly, the set of allowed output words is the \emph{output language} and denoted as $L_{out}(\mathcal{A})$. For an accepted input word, the FST nondeterministically generates a path and gives an output word accordingly, as shown in Example~\ref{ex:fst input-output}. 

\begin{example} \label{ex:fst input-output}
In Figure \ref{fig:observable_fst}, the FST receives $i_1 i_2$ as an input word and can output either $o_1 o_3$ or $o_1 o_2$  as an output word. State 0 is the initial state and state 1 is the final state.
\begin{figure}[h]
  \centering
  \includegraphics[width=0.3\linewidth]{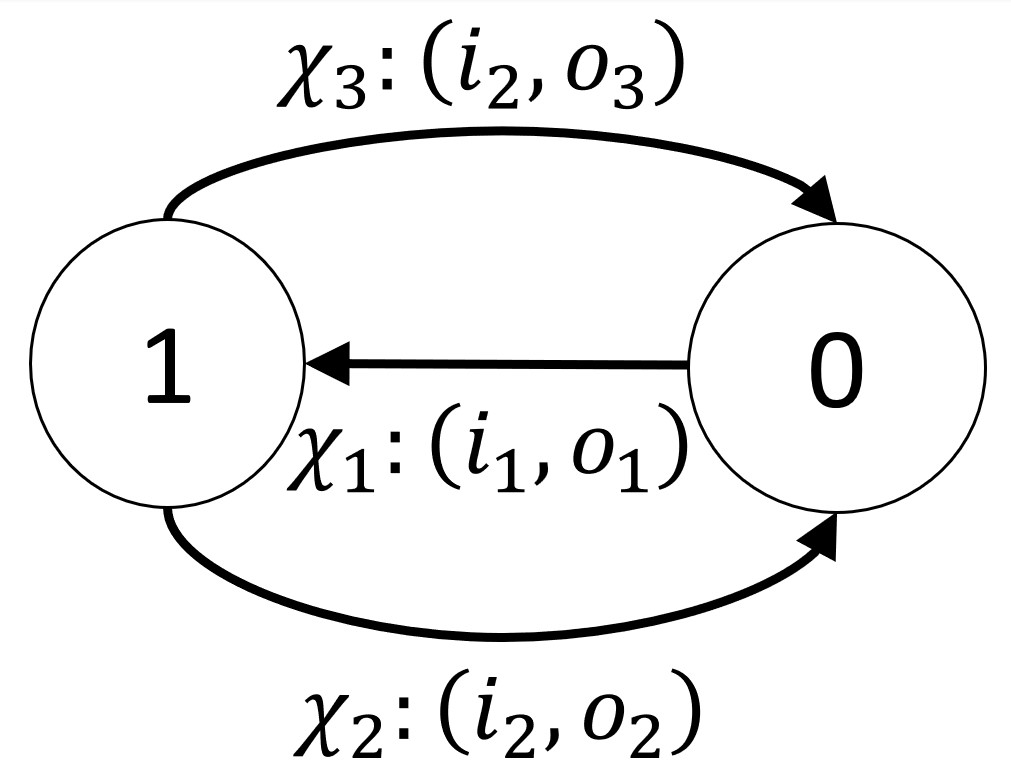}
  \caption{\small An Example FST.}
  \label{fig:observable_fst} 
\end{figure}
\end{example}
\vspace{-16pt}
\subsection{Supervisory control with FSTs}
\label{subsec:prelim_SCL}
Since we use FSTs to capture broad-spectrum attacks, we adopt the FST-based supervisory control framework proposed in~\cite{wang2023attackresilient}. This supervisory control framework can handle any attack defined by a regular language since every regular language over a two-dimensional alphabet is realized by some FST over the same alphabet \cite{holcombe2004algebraic}. As shown in Figure~\ref{fig:composite_fst_control}, the supervisory control loop is described as follows:
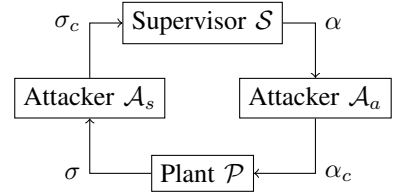
\begin{figure}[!t]
\vspace{8 pt}
\centering
\begin{tikzpicture}
\node (center) {};
\node [draw, right of=center, node distance=1.5cm]  (ai) {Attacker $\mathcal{A}_a$};
\node [draw, below of=center, node distance=1cm]  (p) {Plant $\mathcal{P}$};
\node [draw, left of=center, node distance=1.5cm] (ao) {Attacker $\mathcal{A}_s$};
\node [draw, above of=center, node distance=1cm] (s) {Supervisor $\mathcal{S}$};

\draw [->] (ai) |- node [below,right] {$\alpha_{c}$} (p);
\draw [->] (p) -| node [below,left] {$\sigma$} (ao);
\draw [->] (ao) |- node [above,left] {$\sigma_{c} $} (s);
\draw [->] (s) -| node [above,right] {$\alpha$}  (ai);
\vspace{-50 pt}
\end{tikzpicture}

\caption{Supervisory control loop under sensor and actuator attacks. Sensor attacker manipulates messages from plant to supervisor. Actuator attacker manipulates messages from supervisor to plant}
\label{fig:composite_fst_control}
\end{figure}
\begin{enumerate}
\item The supervisor (modeled by an FST) chooses an output symbol $\alpha$, which will be sent to the plant, based on an eligible transition that has this symbol as its output.  
\item The actuator attacker $\mathcal{A}_a$ receives $\alpha$ and makes a transition whose input symbol matches $\alpha$. There may be many such transitions and the attacker can choose any one non-deterministically. This transition will generate another symbol $\alpha_c$ which will be sent to the plant.
\item The plant receives $\alpha_c$ and non-deterministically makes a transition whose input symbol matches it. This transition will generate an output symbol $\sigma$.
\item The sensor attacker $\mathcal{A}_s$ receives $\sigma$ and non-deterministically makes a transition whose input symbol matches it. This transition will generate an output symbol $\sigma_c$, which will be sent to the supervisor.
\item Upon receiving $\sigma_c$, the supervisor completes this iteration by executing the transition labeled by $(\sigma_c, \alpha)$. Otherwise, when there is no such transition, the supervisor will stop the control loop and raise an alarm. This completes a singular time step. 
\end{enumerate}
Due to the clocked nature of the FST control loop, the empty symbol $\varepsilon$ represents the empty message and is recorded as such. It has been shown in~\cite{wang2023attackresilient} that the language of the supervised plant (i.e., the set of plant's possible words in the control loop) is given by 
$L(\mathcal{P} | \mathcal{S}, \mathcal{A}_s, \mathcal{A}_a) = L( (\mathcal{A}_s \circ \mathcal{S} \circ \mathcal{A}_a )^{-1}) \cap L(\mathcal{P})$, 
where $\circ$ stands for the standard FST composition. We say a supervisor $\mathcal{S}$ is resilient if $L(\mathcal{P} | \mathcal{S}, \mathcal{A}_s, \mathcal{A}_a) = \mathcal{K}$. Example \ref{ex:fst_sup_control_example_w_attack} provides an example supervisory control loop with an actuator attack present and will be used throughout a running example.

The desired language $\mathcal{K}$ contains all admissible paths rather than just specific paths. Some admissible paths may meet certain objective criterion better than other admissible paths.. Suppose it is desired for the plant to follow a specific path defined by $\chi_1 \in (\mathbf{I} \times \mathbf{O} )^*$, but the attackers have instead forced the plant to follow $\chi_2$. If both $\chi_1$, $\chi_2 \in \mathcal{K}$ the supervisor will not shut down the plant. However, if $\chi_2$ does not fall within $\mathcal{K}$, then the supervisor will terminate the process.  

We desire the plant FST to be capable of stopping at any instant regardless of current state. Consequently, all states for each FST in the control loop are assumed final. This ensures each prefix of an allowed word is also allowed within the control loop. This concept is formalized with the prefix closure of a language. The prefix closure $\overline{L}$ of a language $L$ is defined by including all prefixes of all $\chi \in L$. We say $L$ is prefix closed if $L = \overline{L}$. Each FST in the supervisory control loop is assumed to have a prefix closed language

\begin{example}
\label{ex:fst_sup_control_example_w_attack}
Figure~\ref{fig:ex_fst_attack} illustrates the supervisory control of an FST plant with an actuator attacker. Pairs of input-output symbols for the attacker, plant and superivisor are labeled with $\chi, \beta,$ and $\tau$ respectively. The plant is modeled with the bottom FST. It has one state $0$ and two transitions $\beta_1$ and $\beta_2$ with input/output symbols $(\alpha_1, \sigma_2)$ and $(\alpha_2, \sigma_2)$, respectively. The goal is to restrict the plant to message sequences given by $ \mathcal{K} = \{\beta_1 \beta_2\}^* = \{ (\alpha_1, \sigma_2)(\alpha_2, \sigma_2) \}^*$. The attacker is on the right; it replaces $\alpha_3$ with $\alpha_1$ and vice-versa while in state $0$, and $\alpha_3$ with $\alpha_2$ while in state 1. Notice that the actuator attacker has actuation symbols for both input and output symbols. In order for the supervisor to successfully restrict the plant to $\mathcal{K}$, the supervisor must now send the plant $(\alpha_3 \alpha_1)^*$ ($\alpha_3$ followed by $\alpha_1$) to account for the actuator attacker.
\begin{figure}[h]
\vspace{9 pt}
\centering
\includegraphics[width=0.6\linewidth]{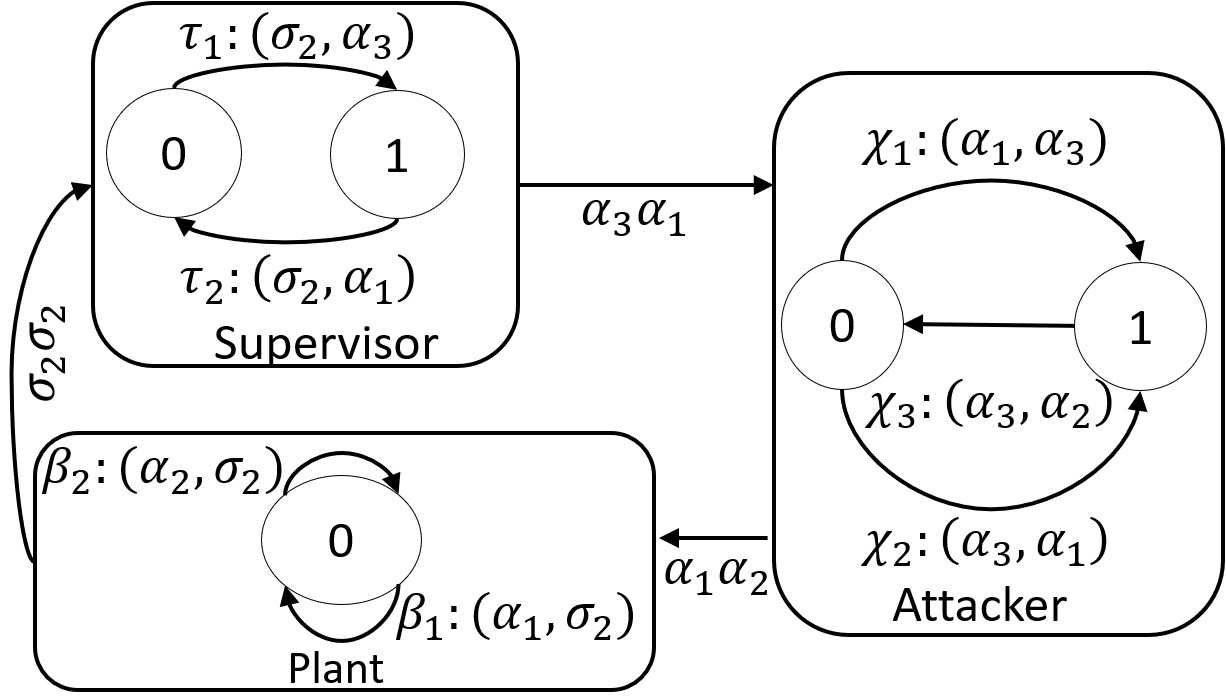}
\caption{\small Example of supervisory control loop with an actuator attacker $\mathcal{A}_a$. \label{fig:ex_fst_attack} }
\end{figure}

\end{example}

\vspace{-20pt}
\section{Problem Formulation}
\label{sec:problem_form}
Our problem formulation studies the resilient supervisor synthesis problem  in the FST framework when the attack models $\mathcal{A}_a$ and $\mathcal{A}_s$ are \emph{unknown}. We assume that both the sensing and actuation channels are subject to attacks. Additionally, we assume that there exist FST models $\mathcal{A}_a$ and $\mathcal{A}_s$ that accurately model these attackers. Our goal is to use recorded attacks from $\mathcal{A}_a$ and $\mathcal{A}_s$ to design a resilient supervisor. 

Following the developments in \cite{wang2023attackresilient}, if adequate models of the unknown attackers can be obtained, then these models can be used to construct a resilient supervisor. This is done by composing the desired language of the plant with the inverse of the attacker FST models to reverse their effects with \begin{equation}
\label{eq:sup_candidate}
    \mathcal{S} = \mathcal{A}_s^{-1} \circ \mathcal{M_K}^{-1} \circ  \mathcal{A}_a^{-1}
\end{equation} 
where $\mathcal{K}  = L( \mathcal{M_K})$.  Equation \eqref{eq:sup_candidate} provides a resilient supervisor from the attacker models when such models are available. If after computing $\mathcal{S}$ with \eqref{eq:sup_candidate} we find that 
\begin{equation}
\label{eq:sup_req}
L(\mathcal{P} | \mathcal{S}, \mathcal{A}_s, \mathcal{A}_a)= \mathcal{K}
\end{equation} then $\mathcal{S}$ is resilient. However, if we find that $L(\mathcal{P} | \mathcal{S}, \mathcal{A}_s, \mathcal{A}_a) \neq \mathcal{K}$ then no such supervisor exists. 

We know a resilient supervisor can be built once sufficiently accurate FST models of the attackers are available, so the core problem is learning these models from recorded attack histories.  More specifically, we aim to construct FSTs $\mathcal{A}_a'$ and $\mathcal{A}_s'$ from $\mathcal{D}_a \subseteq L( \mathcal{A}_a )$ and $\mathcal{D}_s \subseteq L( \mathcal{A}_s )$ such that $L( \mathcal{A}_a') = L( \mathcal{A}_a)$ and $L( \mathcal{A}_s') = L( \mathcal{A}_s)$. This process of learning the FST model $\mathcal{A}$ from sample recordings $\mathcal{D} \subseteq L( \mathcal{A} )$ does not depend on whether the samples are collected from the sensor channel or actuation channel. Consequently, our problem reduces to finding a model $\mathcal{A'}$ from $\mathcal{D} \subseteq L( \mathcal{A} )$ such that $L( \mathcal{A}') = L( \mathcal{A})$. 

\vspace{-8pt}
\section{Data Requirements}
\label{sec:data_requirements}

We will now present a sufficiency condition on $\mathcal{D}$ to ensure it is representative of all possible attack behaviors from $\mathcal{A}$. If $\mathcal{D}$ meets this sufficiency condition, then the FST model of the attacker can be learned. Such a condition involves curating $n+1$ matrices from $\mathcal{D}$ where $n = | \mathbf{I} \times \mathbf{O}| $. This collection of matrices is defined as $H_z$.  The rows and columns in these matrices are indexed by attack sequences recorded in $\mathcal{D} \subset L(\mathcal{A} )$.   


The sufficiency condition involves the rank of such matrices since each linearly independent row within a matrix in $H_z$ is representative of a state in the unknown attacker FST $\mathcal{A}$. More precisely, each linearly independent row is representative of a state within the minimum size FST equivalent to $\mathcal{A}$. If $H_z$ does not contain enough linearly independent rows then $H_z$ is not representative of all possible attack behaviors. This is because the matrices in $H_z$ derived from attack sequences $\mathcal{D}$ can be generated from an FST with states fewer than the minimum state FST representation of $\mathcal{A}$. 




The first matrix that we will consider in $H_z$ is known as $H_\Theta$ and must meet a rank condition. This matrix is defined by two different sets of words and describes what words have been collected in $\mathcal{D}$. This is a binary matrix where 1's represent words that have been collected in $\mathcal{D}$ and 0's for words that have not been collected. In order for the sufficiency condition to be met $H_\Theta$ must not contain any errors e.g. fale positives or false negatives. Otherwise, $H_\Theta$ will represent behaviors that do not match the attacker's true behavior. As previously mentioned, $H_\Theta$ must also meet a rank condition to gaurantee all the behaviors are represented from the unknown attacker. Specifically, $H_\Theta$ must have a rank equal to the state size of $\mathcal{A}$. The rank of $H_\Theta$ is equal to the state size of an FST that would produce the patterns in $H_\Theta$. These conditions and the explicit definition for $H_\Theta$ are given below in \ref{def:basis}. 
\begin{definition}
\label{def:basis}
Two subsets, $\Psi, \Gamma \subset  (\mathbf{I} \times \mathbf{O} )^*$, are needed to define $H_\Theta$; the pair, $\Theta=(\Psi, \Gamma )$, is known as a mask.  Any mask $\Theta = ( \Psi , \Gamma )$ defines $H_{\Theta} \in \mathbb{R}^{|\Psi| \times |\Gamma|}$ with $H_{\Theta}(\psi, \gamma ) = 1$ if and only if $\psi\gamma \in  \mathcal{D}$, and $H_{\Theta}(\psi, \gamma ) = 0$ otherwise for all $(\psi, \gamma ) \in \Psi \times \Gamma$. A mask that meets the conditions to ensure that $\mathcal{D}$ is an adequate representation of the unknown attacker FST is referred to as a basis. Given an FST language $L( \mathcal{A} )$ for some FST $\mathcal{A}$, we call $(\Psi , \Gamma)  = \Theta$ a basis for $L( \mathcal{A} )$ if $ \mathrm{rank}( H_\Theta ) = |\mathbf{S}| $, and $ \varepsilon \in \Psi \cap \Gamma$ where $\mathbf{S}$ is the set of states for $\mathcal{A}$. 
\end{definition}
\vspace{-6pt}
Observe that $\mathcal{D}$ and $H_\Theta$ can satisfy this rank condition from a relatively small sample size allowing us to eventually learn the FST model from a small sample of possible attack behaviors. This allows our method to account for samples not directly observed in $\mathcal{D}$. An example basis will now be given in the context of the running example from Figure \ref{fig:ex_fst_attack}. This will illustrate how one curates $H_\Theta$ from the sample attack set $\mathcal{D}$. 
\vspace{-3pt}
\begin{example}
\label{ex:hankle_basis}
Continuing our running example, consider the following set of recorded words from the actuator attacker FST $\mathcal{A}_a$ depicted in Figure \ref{fig:ex_fst_attack}. In the context of our problem formulation, we only have access to $\mathcal{D}$ below and not the model of the attacker itself.
\begin{equation*} 
\mathcal{D} = 
    \begin{Bmatrix}
    \chi_1, &
    \chi_1 \chi_3,  &
    \chi_1 \chi_3 \chi_1,  &
    \chi_2,  \\ 
    \chi_2 \chi_3,  &
    \chi_2 \chi_3 \chi_2,  &
    \chi_1 \chi_3 \chi_2,  &
    \chi_2 \chi_3 \chi_1 
    \end{Bmatrix}.
\end{equation*}
From these samples, we construct a basis $(\Psi , \Gamma  = \Theta)$ where $\Psi = \{ \varepsilon, \chi_1, \chi_2 \}$ and $\Gamma = \{ \varepsilon, \chi_1, \chi_2, \chi_3 \}$. The matrix subblock $H_\Theta$ is given below. Observe that from only eight sample attacks $\mathcal{D}$ is of sufficient size to extrapolate all possible attack behaviors from the unknown FST attacker.
\[
  H_\Theta = \kbordermatrix{
    & \varepsilon & \chi_1 & \chi_2 & \chi_3 \\
    \varepsilon & 1 & 1 & 1 & 0 \\
    \chi_1 & 1 & 0 & 0 & 1 \\
    \chi_2 & 1 & 0 & 0 & 1 \\
  }
\]

\end{example}
For each $\chi \in \mathbf{I} \times \mathbf{O}$ we must curate a matrix from $\mathcal{D}$ known as $H_\chi$ that provides information regarding possible transitions labeled by $\chi$. While $H_\Theta$ encodes information about the word $\psi \gamma$ with $(\psi ,\gamma) \in \Psi \times \Gamma$.  The Hankel matrix $H_\chi$ encodes information about the word $\psi \chi \gamma$ for each letter $\chi$ in $\mathbf{I} \times \mathbf{O}$. The collection of each $H_
\chi$ and $H_\Theta$ will be denoted as $H_z$. If  $H_z$ meets sufficiency conditions for us to learn $\mathcal{A}$, we say $H_z$ \emph{represents} the unknown attacker FST $\mathcal{A}$. Such conditions require $H_\Theta$ to be of sufficient rank and each $H_\chi$ to contains no errors. A formal definition for $H_z$ and such conditions are given below. 

\begin{definition}
    \label{def:H_z}
Let $\mathcal{A}$ be an FST with alphabet $\mathbf{I} \times \mathbf{O}$ and state set $\mathbf{S}$. Let $n = | \mathbf{I} \times \mathbf{O} |$. Suppose $\Theta = (\Psi,\Gamma)$ is a mask and $\mathcal{D} \subset L( \mathcal{A} )$. Then for any $\chi \in \mathbf{I} \times \mathbf{O}$ we can define a matrix $H_{\chi} \in \mathbb{R}^{|\Psi| \times |\Gamma|}$ with $H_{\chi}(\psi, \gamma) = 1 $ if and only if  $\psi \chi \gamma \in  \mathcal{D}$ and $H_{\chi}(\psi, \gamma) = 0 $ otherwise where $(\psi, \gamma ) \in (\Psi \times \Gamma)$.  The collection of all such matrices is denoted as $H_z = 
\begin{Bmatrix} 
H_{\Theta}, & 
H_{\chi_1}, & 
H_{\chi_2}, & 
\dots, & 
H_{\chi_n}
\end{Bmatrix}$
where $n = | \mathbf{I} \times \mathbf{O} |$. We say $H_z$ \emph{represents} $\mathcal{A}$ if $H_\Theta$ is a basis for $L(\mathcal{A})$, and each $H_{\chi_i}$ contains no false negatives or false positives.
\end{definition}
Intuitively, it helps to think of $H_z$ as a set of matrices that capture all the attack patterns present in $\mathcal{D}$. The rank of $H_\Theta$ represents the state size of an FST responsible for producing the attacks in $\mathcal{D}$. Thus, if the matrix $H_\Theta$ does not have sufficient rank then $\mathcal{D}$ represents an FST that is "smaller" than the true unknown attacker FST. The error condition on each $H_\chi$ requires that one does not record a 0 for a word that belongs in $L(\mathcal{A} )$. This condition guarantees that $H_z$ can generalize all attacks from $\mathcal{A}$. 
Next, we will continue the running example and show how each $H_\chi$ is derived from $\mathcal{D}$ and $\Theta$.
\begin{example}
\label{ex:hankle_subblock}
The matrices derived from the basis used in Example \ref{ex:hankle_basis} are provided below. All matrices are contained in the set $H_z$. The reader is reminded that each entry in the matrices is a 1 if the word is recorded in the sample set and 0 otherwise.
\[
 H_{\chi_1} = H_{\chi_2} = \kbordermatrix{
    & \varepsilon & \chi_1 & \chi_2 & \chi_3 \\
    \varepsilon & 1 & 0 & 0 & 1 \\
    \chi_1 & 0 & 0 & 0 & 0 \\
    \chi_2 & 0 & 0 & 0 & 0 \\
  }
\]
\[
 H_{\chi_3} =  \kbordermatrix{
    & \varepsilon & \chi_1 & \chi_2 & \chi_3 \\
    \varepsilon & 0 & 0 & 0 & 0 \\
    \chi_1 & 1 & 1 & 1 & 0 \\
    \chi_2 & 1 & 1 & 1 & 0 \\
  }
\]
\begin{equation*}
    H_z = \{H_\Theta,  H_{\chi_1}, H_{\chi_2}, H_{\chi_3} \} 
\end{equation*}
\end{example}



In practice, one will choose $\Psi$ and $\Gamma$ as a mask that maximizes the rank of $H_{\Theta}$; this will create the most representative model of the true attacker. Furthermore, if for any $\chi \in \mathbf{I} \times \mathbf{O}$, $H_{\chi}$ is not in the row space of $H_{\Theta}$ then $\Theta$ is not a basis for $L(\mathcal{A})$ \cite{carlyle_realizations_1971}. In this case, increase the size of $\Theta$ or collect more samples. The enforcement of this property is equivalent to the closed property of an observation table described in $L^*$ algorithm \cite{angluinLearningRegularSets1987}. 

To conclude this section, we introduce an informal theorem linking Hankel matrices to an FST model. Specifically, the theorem states that if the unknown attacker $\mathcal{A}$ is sampled adequately so that its behavior is captured by the collection $H_z$ of Hankel matrices, then matrix algebra can be used to construct an equivalent finite state transducer $\mathcal{A}'$.

\begin{informaltheorem}
\label{theorem:informal}
Assume $\mathcal{A}$ has been sufficiently sampled to obtain a collection of matrices $H_z = \{ H_\Theta, H_{\chi_1}, ..., H_{\chi_n} \}$ that represent $\mathcal{A}$. Then, matrix algebra on $H_z$ provides an FST $\mathcal{A}'$ that is equivalent to $\mathcal{A}$.
\end{informaltheorem}

\vspace{-12pt}
\section{Spectral learning for FSTs}
\label{sec:spec_method}



This section presents a novel method for deriving attack models as finite state transducers (FSTs) from the set of Hankel matrices constructed from attack sample data in Section \ref{sec:data_requirements}. To bridge the gap between these Hankel matrices and the FST's graph structure, we begin by introducing the concept of a transition tuple (a matrix representation of an FST) in Section \ref{subsec:trans_tuples}. Next, we show how applying a matrix spectral analysis method based on the singular value decomposition (SVD) of such Hankel matrices yields a transition tuple for the unknown attacker model in Section \ref{subsec:SVD2tuples}. Finally, we develop a new methods that maps this transition tuples to the FST model of the attacker in Section \ref{subsec:tuples2FST}. The key innovation in our approach is the use of a linear transformation to convert the SVD into a decomposition consisting of basis vectors we denote as \emph{natural}. 

\vspace{-12pt}
\subsection{Transition Tuples}
\label{subsec:trans_tuples}

In this subsection, we explain how a transition tuple captures an FST in matrix form and the relationship between the two representations. When the transition tuple is in a form we call \emph{natural}, the corresponding FST model can be easily derived (as formalized in Lemma \ref{lemma:fst_def_equiv}). A transition tuple consists of a set of transition matrices, each associated with a symbol from the FST's alphabet, where the rows and columns of each matrix are indexed by the FST states. We say that the transition tuple \textit{realizes} a language if the product of these matrices reproduces the FST's language. 


We now define transition tuples explicitly. For a given FST $\mathcal{A}$, a transition tuple provides the language $L( \mathcal{A} )$, but not necessarily the graph structure within Definition \ref{def:fst}. If the transition tuple can also directly provide $\mathcal{A}$ through Definition \ref{def:fst} then we denote it as \emph{natural}.  For a vector $t \in \mathbb{R}^n$, we denote by $t^T$ its transpose, and the standard basis row vectors for $\mathbb{R}^n$ are given by $e_i$ where $i \in \{1, 2, \dots n\}$. 

\begin{definition}
\label{def:trans_tuple}
For an FST alphabet $\mathbf{I} \times \mathbf{O}$, we say $( t_0,  t_{\infty}, \{ T_{\chi} \}_{\chi \in \mathbf{I} \times \mathbf{O}} )$ is a transition tuple $\mathcal{T}$, if
\begin{itemize}
    \item  $t_\infty, t_0  \in \mathbb{R}^n$ ;
    \item  $T_\chi \in \mathbb{R}^{n \times n}$ for all $\chi \in \mathbf{I} \times \mathbf{O} $;
    \item The mapping  defined by $f_\mathcal{T}( \chi_1\chi_2...\chi_n ) = t_0^T T_{\chi_1} T_{\chi_2} \dots T_{\chi_n} t_{ \infty } \in \{0,1\}$ for any sequence in $(\mathbf{I} \times \mathbf{O})^*$. 
\end{itemize}
With a slight abuse of notation for any $ \chi_1 \chi_2 \dots \chi_n = \chi \in ( \mathbf{I} \times \mathbf{O} )^* $ we define $T_\chi = T_{\chi_1} T_{\chi_2} \dots T_{\chi_n}$. We say a transition tuple $\mathcal{T}$ realizes an FST language $L$ when $f_\mathcal{T}( \chi ) = 1$ if and only if $\chi \in L$. A matrix $T$ is \emph{natural} if each row in $T$ is $e_i$ for some $i \in \mathbf{N}$ or the zero vector. A transition tuple $\mathcal{T}$ is \emph{natural} if $t_0$, and each $T_\chi$ are natural while $t_{\infty}$ is binary. 
\end{definition} 

If $\mathcal{T}$ is a transition tuple that realizes the language of an FST $\mathcal{A}$ then the map defined by $f_\mathcal{T}$ provides this language. The function $f_\mathcal{T}:(\mathbf{I} \times \mathbf{0})^* \rightarrow \mathbb{R}$  is defined by the composition of the matrices within the tuple $\mathcal{T}$. The matrix composition order used to compute $f_\mathcal{T}(\chi)$ matches the letter order within the word $\chi$.  If the word $\chi$ is accepted by the FST $\mathcal{A}$ and $\mathcal{T}$ realizes $L(\mathcal{A})$ then $f_\mathcal{T}(\chi)=1$. Conversely, when $\chi$ is not accepted by the FST $\mathcal{A}$ then $f_\mathcal{T} ( \chi ) = 0$.  

In order for a transition tuple to explicitly provide the FST model itself with Definition \ref{def:fst} the transition tuple must be \emph{natural}. More precisely, if $\mathcal{T}$ is a transition tuple that realizes $L(\mathcal{A})$ for some FST $\mathcal{A}$ , then the map $f_\mathcal{T}$ can determine which words belong to $L( \mathcal{A} )$ and which words do not. Nonetheless, the transition tuple defined with Definition \ref{def:trans_tuple} may not provide the initial state, final state, and transition relation needed to provide $\mathcal{A}$ with \ref{def:fst}. A \emph{natural} transition tuple provides the FST structure directly. We will now show how a natural transition tuple provides the FST structure directly with Lemma \ref{lemma:fst_def_equiv}.

\begin{lemma}
\label{lemma:fst_def_equiv}
A natural transition tuple $\mathcal{T}$ defines an FST $\mathcal{A}$ , and $\mathcal{T}$ realizes $L( \mathcal{A} )$. Similarly, given an FST $\mathcal{A}$, there exists a natural transition tuple that realizes $L( \mathcal{A} )$. 
\end{lemma}

\begin{proof}
\label{proof:trans_tuple_realize_A}
Let $\mathcal{T} = ( t_0,  t_{\infty}, \{ T_{\chi} \}_{\chi \in \mathbf{I} \times \mathbf{O}})$ be a natural transition tuple, where $t_0,  t_{\infty} \in \mathbb{R}^{n}$ and $T_{\chi} \in \mathbb{R}^{n \times n}$ for each $\chi \in \mathbf{I} \times \mathbf{O}$. Then, we can construct the FST $\mathcal{A} = (\mathbf{S}, \mathbf{s_0}, \mathbf{I}, \mathbf{O}, \mathbf{Trans}, \allowbreak \mathbf{S_f})$ as follows: 
\begin{itemize}
\item $\mathbf{S} = \{\mathbf{s}_1, \ldots, \mathbf{s}_n\}$.
\item $\mathbf{s_0} = s_k$ if and only if $t_0(k) = 1$, i.e., the $k$ entry of $t_0$ being $1$. It is unique since $t_0$ is natural.
\item For each $(i, o) \in  \mathbf{I} \times \mathbf{O}$,
$(s_p, i, o, s_q) \in \mathbf{Trans}$ if and only if $T_{(i, o)} (p, q) = 1$, i.e., the $(p,q)$ entry of $T_{(i, o)}$ being $1$. 
\item $s_k \in \mathbf{S_f}$ if and only if $t_\infty(k) = 1$. It is well-defined since $t_\infty$ is binary.
\end{itemize}
By construction, the transition tuple $\mathcal{T}$ realizes the same language as the FST $\mathcal{A}$.  An identical argument is made to show that some natural transition tuple realizes the language of any FST. Let $\mathcal{A}$ be an FST and $\mathcal{A}'$ be an FST such that $L(\mathcal{A}) = L( \mathcal{A}' )$ with $\mathcal{A'} = (\mathbf{S}', \mathbf{s_0}', \mathbf{I}, \mathbf{O}, \mathbf{Trans}', \allowbreak \mathbf{S_f'})$ where $\mathbf{Trans'}$ is a partial function from $\mathbf{S} \times \mathbf{I} \times \mathbf{O}'$ into $\mathbf{S}'$. The existence of $\mathcal{A'}$ is shown in \cite{wang2023attackresilient}. 
Without loss of generality, assume $\mathbf{S}' = \{1,2, ..., |\mathbf{S}|'\}$. Next define a Transition Tuple $\mathcal{T} = ( t_0,  t_{\infty}, \{ T_{\chi} \}_{\chi \in \mathbf{I} \times \mathbf{O}} )$ as 
\begin{itemize}
\item For each $\chi \in \mathbf{I} \times \mathbf{O}$ set $T_{\chi} \in \mathbb{R}^{| S' | \times | S' | }$ with $e_{ p }^T T_{ \chi } = e_{ q } $ if $(p, i, o, q) \in \mathbf{Trans'}$ where $\chi = (i,o)$ and $e_{ p }^T T_{ \chi } = 0 $ otherwise. 
\item $t_0 \in \mathbb{R}^{| \mathbf{S}' |} $ is a column vector with $t_0 = e_{ s_0 }$;
\item $t_{\infty} \in \mathbb{R}^{ | \mathbf{S}' | }$ is a column vector with $e_{ s }^T t_{\infty} = 1$ if $s \in \mathbf{S_f}' $. Otherwise, $e_{ s }^T t_{\infty} = 0$.
\end{itemize}
We will show that for any $\chi \in (\mathbf{I} \times \mathbf{O})^*$, $\chi \in L(A)$ iff $f_\mathcal{T}( \chi ) = 1$. First, suppose $\chi = \chi_1\chi_2 \dots \chi_n \in L( \mathcal{A} )$ and let $\pi = \{ (s_{k}', i_k, o_k, s_{k+1}' )\}^n_{k = 1}$ be the unique path associated with $\chi$ in $\mathcal{A}'$. Since $\pi$ is a valid path observe that for all $k \in \{0,1, \dots, n\}$, we have $ e_{ s_k' }^T T_{ \chi } = e_{ s_{k+1}' } $, and $e_{ s_{n+1}' }^T t_{\infty} $ = 1. Moreover, $t_0 T_{ \chi^m } = e_{ s_{ m+1 }  }$ where $\chi^m =\chi_1 \dots \chi_m $ because $t_0 = e_{ s_0  }$. Thus, 
\begin{equation*}
    f_\mathcal{T} (\chi) = t_0^T T_\chi t_{\infty} = e_{ s_{ n+1 }  } t_{\infty} = 1.
\end{equation*}
Now consider the case where $\chi = \chi_1\chi_2 \dots \chi_n \not \in L( \mathcal{A} )$. Then any path $\pi = \{ (s_{k}, i_k, o_k, s_{k+1} )\}^n_{k = 1}$ with $\chi = (i_1, o_1 )...(i_n, o_n)$ must fail to satisfy one of two properties. Each path $\pi$ makes an invalid transition or does not end in the final state. In both cases, $f_\mathcal{T} ( \chi ) = 0$ because either $ e_{ s_k }^T T_{ \chi } = 0 \neq e_{ s_{k+1} } $ for some $k$, or $ e_{ s_{k+1}  }^T t_{\infty} = 0 $ respectively. Thus,  $\mathcal{T}$ realizes $L( \mathcal{A} ) = L( \mathcal{A}' )$ . 
\end{proof}
Lemma \ref{lemma:fst_def_equiv} shows the equivalence between a natural transition tuple that realizes $L( \mathcal{A} )$ and the FST $\mathcal{A}$ itself. If $\mathcal{T}$ is a natural transition tuple that realizes $L(\mathcal{A})$ for some FST $\mathcal{A}$ then we can easily obtain the initial and final states along with the transition relation defining $\mathcal{A}$. For example, each matrix $T_\chi$ in the nautral transition tuple provides all possible single-step transitions labeled by $\chi \in \mathbf{I} \times \mathbf{0}$. Similarly, $t_0$ provides the initial state and $t_\infty$ provides the final states for the FST $\mathcal{A}$. 

\subsection{Deriving transition tuples from Hankel matrices with SVD}
\label{subsec:SVD2tuples}
In this subsection, we demonstrate how spectral learning transforms a set of Hankel matrices into a transition tuple $\mathcal{T}$ that represents the unknown attacker FST model $\mathcal{A}$. Each Hankel matrix in the set $H_z$ encodes information about sample attacks that involve sequences of multiple letters. In contrast, the transition tuple $\mathcal{T}$ provides a breakdown of the attacker model on a per-letter basis, assigning a transition matrix to each letter in the alphabet. Singular value decomposition (SVD) can rearrange the matrices within $H_z$ to provide information about $\mathcal{A}$ on a per-letter basis. The transition tuple $\mathcal{T}$ can then be derived from such a rearrangement. This process is formalized in Lemma \ref{lemma:factorization2trans_tuple}.



Each transition matrix $T_\chi$ is computed using the SVD of both the basis Hankel matrix $H_\Theta$ and the Hankel matrix $H_\chi$ corresponding to letter $\chi$. Although $H_\chi$ contains essentially the same word information as $H_\Theta$, every word in $H_\chi$ includes $\chi$. This difference allows us to combine the SVD results to isolate the model information specific to $\chi$. To help clarify this extraction process, we introduce the $\mathrm{pre}$ and $\mathrm{post}$ functions, whose formal definitions are provided in Definition \ref{def:pre/post}.


\begin{definition}
    \label{def:pre/post}
     Given any transition tuple $\mathcal{T} = ( t_0,  t_{\infty}, \allowbreak \{ T_{\chi} \}_{\chi \in \mathbf{I} \times \mathbf{O}} )$ define  $\mathrm{pre} , \mathrm{post} : ( \mathbf{I} \times \mathbf{O} )^* \rightarrow \mathcal{R}^n$ as 
\begin{equation}
\label{eq:pre_post_def}
    \mathrm{pre} ( \chi ) = t_0^T T_{\chi} \text{ and } \mathrm{post} ( \chi ) =  T_{\chi} t_{\infty}.
\end{equation}
When $\mathrm{pre}(\chi_1) \mathrm{post} (\chi_2) = f_\mathcal{T}( \chi )$ for any $\chi_1 \chi_2 = \chi$ within $( \mathbf{I} \times \mathbf{O} )^* $ we say $\mathrm{pre}$, $\mathrm{post}$ realize $f_\mathcal{T}$. Note that $\mathrm{pre} (\chi)$ is a row vector and $\mathrm{post} (\chi)$ is a column vector.
\end{definition}
Consider $\chi_1, \chi_2$ and $\chi_3$ from Example \ref{ex:fst input-output} and let $\mathcal{T}$ be a transition tuple that realizes $L(\mathcal{A})$. Observe that we have 
\begin{equation}
\label{eq:pre_post_intuit}
\begin{split}
H_{\chi_3} (\chi_1, \chi_2 )
& = f_\mathcal{T}( \chi_1 \chi_3 \chi_2 )  = t_0 T_{\chi_1} T_{\chi_3} T_{\chi_2} t_{\infty}^T  
\\ & = \mathrm{pre}(\chi_1) T_{\chi_3} \mathrm{post}(\chi_2)
\end{split}
\end{equation} 
where the first, second, and third equalities fall immediately from Definitions \ref{def:H_z}, \ref{def:trans_tuple}, and \ref{def:pre/post}  respectively. Furthermore, if $\mathrm{pre}(\chi_1)$ and $\mathrm{post}(\chi_2)$ were known then one could solve for $T{ \chi_3 }$ with
\begin{equation*}
\mathrm{pre}(\chi_1) T_{\chi_3} \mathrm{post}(\chi_2) = 1
\end{equation*}
because $\chi_1 \chi_3 \chi_2 \in \mathcal{D}$. To compute $\mathrm{pre}(\chi_1)$ and $\mathrm{post}(\chi_2)$,  we note that $H_\theta(\chi_1, \chi_2) =     \mathrm{pre}(\chi_1)  \mathrm{post}(\chi_2)$. This relationship extends over the entire basis $(\Psi, \Gamma)$. In other words, if we know the sets $\mathrm{pre}(\Psi)$ and $\mathrm{post} (\Gamma)$, then the relation in \eqref{eq:pre_post_intuit} allows us to determine each transition matrix $T_\chi$. Moreover, decomposing $H_\Theta$ yields the desired evaluations of $\mathrm{pre}$ and $\mathrm{post}$ over the basis. We use SVD to obtain matrices $P$ and $S$, which have full column and full row rank respectively, so that $H_\theta = PS$. If $H_\Theta = \Sigma \Omega V'$ is the SVD then by convention $P= \Sigma U$ and $S=V'$. Here, $U$and $V'$ contain the singular vectors and $\Omega$ contains the singular values. In this decomposition, $P$  and $S$ contains the evaluations of $\mathrm{pre}(\Psi)$ and $\mathrm{post} (\Gamma)$ respectively. 

Now, that we have $\mathrm{pre}, \mathrm{post}$ evaluated over $\Psi$ and $\Gamma$ respectively and $H_{\chi_1}, H_{\chi_2}, H_{\chi_3}$ we can solve for $T_{\chi_1}, T_{\chi_2}$ and $T_{\chi_3}$. We can extrapolate \eqref{eq:pre_post_intuit} over any $\chi$ within the attacker alphabet to obtain $H_{\chi} = PT_{\chi}S$.
Lastly, the pseudoinverses of $P$ and $S$ can solve for the transition tuple $\mathcal{T}$ with
\begin{equation}
\label{eq:trans_solver}
   T_{\chi} = P^L H_{\chi} S^R \text{ ,  } t_0 = P(1,:) \text{ ,  } t_{\infty} = S(:,1).
\end{equation}
where $P^L$ and $S^R$ denote the left and right pseudoinverse of $P, S$ respectively.  It is important to note that \eqref{eq:trans_solver} only holds if the selected mask is a basis and $P, S$ are full column and row rank respectively. Lemma \ref{lemma:factorization2trans_tuple} below formally describes how we obtain the language of the unknown FST through sufficient samples using SVD and matrix algebra on $H_z$; the language of the unknown FST is in the form of a transition tuple that realizes such a language.

\begin{lemma}
    \label{lemma:factorization2trans_tuple}
Let $\mathcal{A}$ be an FST and assume $\Theta$ is a basis for $L( \mathcal{A} )$. Given $(P, S)$ as a full rank decomposition of $H_{\Theta}$, and $H_z$ then \eqref{eq:trans_solver} provides a transition tuple $\mathcal{T}$ such that $\mathcal{T}$ realizes $L( \mathcal{A} )$. 
\end{lemma}

\begin{proof}
Let $\mathcal{A}$, $\mathcal{T}$, $\Theta$, $(P, S)$, and $H_z$ be defined as above. From Lemma \ref{lemma:fst_def_equiv} we know that such a transition tuple exists by construction. The proof follows immediately from Lemma 2 in \cite{balle_spectral_2014} where we equate the FST alphabet to a weighted finite automaton alphabet and set the weight of every possible transition to 1. 
\end{proof}

\vspace{-12pt}
\subsection{Natural Decomposition for FST Model procurement }
\label{subsec:tuples2FST} 

This section presents a novel process for extracting an FST model for an unknown attacker from the transition tuple derived in Section \ref{subsec:SVD2tuples}. Our innovative approach adjusts the SVD of the Hankel matrix $H_\Theta$ so that the resulting transition matrices provides the FST model structure. More precisely, since each transition matrix $T_\chi$ is computed with SVD of $H_\Theta$ and $H_\chi$ we can adjust $H_\Theta$ to ensure each $T_\chi$ is natural providing the FST model. First, we will present how adjusting the SVD of $H_\theta$ through a linear transformation will adjust the transition tuple by the same transformation. Then, we derive a linear transformation that when applied to the SVD of $H_\theta$ forces the resulting transition tuple to be natural. Such a procedure has yet to be developed until now.

\subsubsection{Equivalent Condition for natural transition tuple}
\label{subsubsec:equiv_determ}
Adjusting the basis decomposition, $H_\Theta = PS$, affects the form of the transition tuple obtained with \eqref{eq:trans_solver}. In particular if we transform $P,S$ via linear transformation then the subsequent transition tuple is transformed by the same linear transformation. This is made evident through \eqref{eq:trans_solver}. To illustrate, let $B$ be an invertible matrix of appropriate size, $H_\Theta = PS$ be a basis decomposition of full column and row rank respectively and $H_z$ contain the subsequent Hankel matrices defined with $H_\Theta$. Then $P_{new} = PB$, and $S_{new}= B^{-1}S$ will also provide a full rank decomposition and lead to the following transition tuple $\mathcal{T}^{new} = ( t_0^B,  t_{\infty}^B, \{ T_{\chi}^B \}_{\chi \in \mathbf{I} \times \mathbf{O}})$ 
\begin{equation*}
   T_{\chi}^{B} = 
   B^{-1}P^L H_{\chi} S^R B = 
   P_{new}^L H_{\chi} S_{new}^R, 
\end{equation*}
\begin{equation*}
t_0^{B} = 
P(1,:)B =  
P_{new}(1,:), 
\end{equation*}
\begin{equation*}
    t_{\infty}^{B} = 
    B^{-1}S(:,1) = 
    S_{new}(:,1).
\end{equation*}
Although the $\mathcal{T}^{new}$ has different matrices it still realizes the same language. Thus, the decomposition of $H_\Theta$ can directly affect the transition tuples. 



If the basis decomposition satisfies a condition we call \textit{natural}, then the transition tuple obtained via \eqref{eq:trans_solver} will also be natural. In other words, when the decomposition consists primarily of basis vectors, the resulting matrix compositions yield binary matrices, leading to a natural transition tuple. We now provide a formal definition of a natural decomposition. 
\begin{definition}
\label{def:natural_fact}
A decomposition $(P, S)$ is \emph{natural} if each row in $P$ is $e_i$ for some $i \in \mathbf{N}$ or the zero vector, and $S$ is binary. 
\end{definition}
Now, Lemma \ref{lemma:equiv_binary} demonstrates how enforcing a natural decomposition guarantees that the transition tuple will be natural.  We accomplish this by exploiting the interdependence between the $\mathrm{pre}$, $\mathrm{post}$ functions and basis decomposition shown in \eqref{eq:pre_post_intuit}.  

\begin{lemma}
\label{lemma:equiv_binary}
Let $\mathcal{A}$ be an FST. Given a basis $\Theta$ for $L( \mathcal{A} )$ and ($P, S$) a full rank decomposition on $H_\Theta$. The transition tuple $\mathcal{T} = (t_0,  t_{\infty}, \{ T_{\chi} \}_{\chi \in \mathbf{I} \times \mathbf{O} })$ derived from the above decomposition with \eqref{eq:trans_solver} is natural if and only if the pair $(P, S)$ induces a natural decomposition on $H_{\Theta}$. 
\end{lemma}


\begin{proof}
($\Longrightarrow$) Let $\mathcal{A}$, $\mathcal{T}$, and $\Theta$ be as defined above and $\mathcal{A}'$ be observable while realizing the same language as $\mathcal{A}$. Assume $\mathcal{T}$ is natural. Hence, for any $\chi \in (\mathbf{I} \times \mathbf{O})^* $ we have that $\mathrm{pre} (\chi)$ will be a natural vector and $\mathrm{post} (\chi)$ will be binary because these functions are all compositions of the matrices above. Therefore, $(P, S)$ is natural because it is simply the concatenation of $\mathrm{pre}, \mathrm{post} $ over $\Theta \subset ( \mathbf{I} \times \mathbf{O} )^* \times ( \mathbf{I} \times \mathbf{O} )^*$. (see \eqref{eq:pre_post_intuit}). Furthermore, such a decomposition exists because Lemma \ref{lemma:fst_def_equiv} showed that any FST can be realized by some natural transition tuple.

($\Longleftarrow$) Now, suppose $(P, S)$ is a natural decomposition on $H_\Theta$. We will proceed using contradiction. If $t_0 \neq e_j$ for some $j \in \mathbf{N}$ then the first row of $P$ will not be natural. Similarly, if $t_{\infty}$ is not binary then the first column of $S$ will not be binary. Thus, in either case, the decomposition $(P, S)$ is not natural. 

Lastly, assume that there exist $\chi \in \mathbf{I} \times \mathbf{O} $ and $j \in \mathbf{N}$ such that the j-th row of $T_{\chi}$ is not natural and that $\chi \in \Phi$. The last assumption simply requires that we have observed $\chi$ while sampling the FST. Let $\chi_n$ and $j$ be as such. Let $\chi_1 \dots \chi_k \chi_n \in L( \mathcal{A}')$ where each $T_{\chi_s}$ with $1 \leq s \leq k$ is natural. In other words, $\chi_n$ is the first reachable letter in $\mathcal{A}'$ with a matrix that is not natural. Such a word exists by assumption.

Observe that if $\chi_1 \dots \chi_k \chi_n \in L(\mathcal{A})$ then $\mathrm{pre} (\chi_1 \dots \chi_k) = e_j$. This is because their transition matrices and $t_0$ (1st row of P) are natural. Thus, the row of $P$ given by $ \mathrm{pre}( \chi_1 \dots \chi_k \chi_n) =\mathrm{pre} (\chi_1 \dots \chi_k)T_{\chi_{n}}$  will be unnatural because it is simply the $j$th row of $T_{\chi_{n}}$. This violates our natural decomposition assumption. Therefore, the transition tuple must be natural.
\end{proof}

\subsubsection{FST Procurement}
\label{subsubsec:fst_procurement}

We now demonstrate how to derive a natural decomposition of the basis matrix $H_\Theta$ that ultimately yields a natural transition tuple. First, we show that the linearly independent rows obtained from the SVD-based decomposition link the SVD with the natural decomposition. Then, we use a linear transformation based on these independent rows to convert the original decomposition into a natural one. Once we show how to obtain a natural decomposition we combine these results into Theorem \ref{theorem:summary} to summarize how attacker data samples can be mapped to a representative FST model.

To analytically enforce a natural decomposition on $H_z$, we show that such a decomposition can be achieved through a linear transformation. First, we prove that any two full-rank decompositions are related by a linear transformation. This establishes that the natural decomposition and the SVD-based decomposition are algebraically connected via a linear transformation. Exploiting this connection, enables the direct computation of this transformation. The rows of this transformation matrix are already present in the original decomposition because the left matrix in the natural decomposition matrix consists of basis vectors. These results are formalized in Lemma \ref{lemma:natural_fact}.
\begin{lemma}
    \label{lemma:natural_fact}
    If $\Theta$ be a basis for an FST $\mathcal{A}$ and $H_\Theta = 
    PS$ is a full rank decomposition then there exists a transformation $B$ such that $P_{new} = PB^{-1} \text{ ,  } S_{new} = BS$ define a natural decomposition on $H_\Theta$. The transformation $B$ is found by extracting linearly independent rows, $p_i$, of $P$. Then vertically concatenating the row vectors into a square matrix, $B$. In other words, set 
\begin{equation}
\label{eq:b_construct}
    B = 
    \begin{bmatrix}
        p_1^T &  
        p_2^T &  
        \dots &  
        p_r^T
    \end{bmatrix}^T
\end{equation}
where $r$ is the rank of $H_{\Theta}$. 
\end{lemma}
 \begin{proof}
          Let $\Theta$ and $\mathcal{A}$ be as such. By Lemma \ref{lemma:fst_def_equiv} there exists a NTT that realizes $L( \mathcal{A} )$. Lemma \ref{lemma:factorization2trans_tuple} showed that a decomposition on $H_\Theta$ derived from an NTT will also be natural. Hence, such a decomposition exists. 
     Let $ (P, S)$ and $(P_{new}, S_{new})$ be full rank decompositions on $H_\Theta$. First, we will show $P_{new} = PA$ and $S_{new} = BS $ for some $B, A$ because $P, P_{new}$ span the same subspace of $\mathbf{R}^n$. Then, $AB=I$ will be derived from the algebra below. By definition, we have 
\begin{equation}
\label{eq:factor}
H_\Theta = P_{new} S_{new} = PS.
\end{equation}
Since the columns in $P, P_{new}$ form bases for the column space of $H_\Theta$ and share the same rank, we have that $P, P_{new}$ generate the column space of $H_\Theta$.  Hence, each column in $P_{new}$ can be expressed as a linear combination of the columns in $P$ and vice-versa. This combination is unique because $P, P_{new}$ are both full column ranks. Therefore, there exists a unique matrix $A$ such that $P_{new} = PA$. An identical argument with the row space of $H_\Theta$ shows the unique existence of $B$ with $S_{new} = BS $. Now, \eqref{eq:factor} can be rewritten as $H_\Theta = P_{new} S_{new} = P AB S = PS$. This implies that $AB=I$, $A = B^{-1}$, and uniqueness immediately follows. This is easily extended to the special case where $(P_{new}, S_{new})$ is a natural decomposition. 

Now let $B$ be such that $H_\Theta = P_{new} S_{new} = P B^{-1}B S$. Observe that each row in $P_{new}$ is a standard basis vector, $e_i$. Hence, each row in $P$ is one of the $r$ rows in $B$ because $P_{new} B = P$. Thus, the matrix $P$ is entirely comprised of the rows in $B$ and vice-versa. Moreover, there are only $r$ different rows in $P$. The rest are repeats. We have now proved that \eqref{eq:b_construct} yields a natural decomposition. This completes the proof. 
 \end{proof}
 
Now, we continue the running example to show how the natural decomposition is found and the subsequent transition tuple from \eqref{eq:trans_solver} is natural.
\begin{example}
\label{ex:COB_matrix}
Returning to our running example introduced in Example \ref{ex:fst_sup_control_example_w_attack}, the transformation, $B$, used to enforce determinism on $(P, S)$ is obtained with the rows of $P$ as 
\begin{equation*}
B = 
    \begin{bmatrix}
         -1.29 & 1.15 \\ 
         -1.29 & -0.58 
    \end{bmatrix}.
\end{equation*}
This transformation provides the following new decomposition;
\begin{equation*}
 P_{new} = PB^{-1}
    \begin{bmatrix}
 1 & 0 \\ 
 0 & 1 \\
 0 & 1
    \end{bmatrix}
, 
 S_{new} = BS = 
\begin{bmatrix}
1 & 1 & 1  & 0 \\ 
1 & 0 & 0 & 1
\end{bmatrix}.
\end{equation*}


Finally, we can apply \eqref{eq:trans_solver} to determine the transition tuple for every $\chi \in \mathbf{I} \times \mathbf{O} $. The transition tuple defining the actuator attacker is found to be 
\begin{equation*}
 T_{\chi_1 } = T_{\chi_2 } = P_{new}^L H_{\chi_{1,2}} S_{new}^R 
 =
    \begin{bmatrix}
 0 & 1 \\ 
 0 & 0 
    \end{bmatrix},
\end{equation*}
\begin{equation*}
T_{\chi_3 } = P_{new}^L H_{\chi_{3}} S_{new}^R 
=
    \begin{bmatrix}
 0 & 0 \\ 
 1 & 0 
    \end{bmatrix}
    ,
\end{equation*}
\begin{equation*}
    t_0 =   
    \begin{bmatrix}
             1 & 0 
    \end{bmatrix}
     \text{ ,  }   
     t_{\infty} =   
    \begin{bmatrix}
             1 \\ 
             1 
    \end{bmatrix}
\end{equation*}
The attacker FST structure shown in Figure \ref{fig:ex_fst_attack} is identical to the FST obtained from the tuples above with the construction process given in Lemma \ref{lemma:fst_def_equiv}. This would not be possible without the change of basis transformation, $B$. Note that some state labels can be "reshuffled" during this process, but the FST structure will not change.
\end{example}

We finish by combining these findings with the encompassing theorem that presents a technique for deriving attacker FST models from a sufficient set of matrices $H_z$.
\begin{formaltheorem}
    \label{theorem:summary}
Let $\mathcal{A}$ be an FST. Assume $\mathcal{A}$ has been sufficiently sampled to obtain $H_z$ that represents $\mathcal{A}$. then Equations \eqref{eq:trans_solver} transformed by the change of basis matrix in Lemma \ref{lemma:natural_fact} provide an FST $\mathcal{A}'$ that is equivalent to $\mathcal{A}$.
\end{formaltheorem}

\begin{proof}

Let $\mathcal{A}$ and $H_z$ be as above. Let $(P, S)$ be a full rank decomposition of $H_\Theta$ from SVD. Adujst $(P,S)$ with \eqref{eq:b_construct} to obtain a natural decomposition $(P_{new}, S_{new})$. In Lemma \ref{lemma:natural_fact} we showed that the transition tuple $\mathcal{T}$ obtained with \eqref{eq:trans_solver} will be a natural transition tuple that realizes $L(\mathcal{A})$. Finally, Lemma \ref{lemma:fst_def_equiv} proved that the FST model $\mathcal{A}'$ constructed from $\mathcal{T}$ is equivalent to $\mathcal{A}$. 
\end{proof}

Finally, we outline how our spectral learning method can be paired with
Equation~\eqref{eq:sup_candidate}. This shows how a resilient supervisor can be constructed from sufficiently many samples generated by unknown attacker FSTs. The resulting guarantee is formalized in the following Corollary.

\begin{corollary}
\label{theorem:resilient_sup}
Assume unknown unknown actuator and sensor attackers $\mathcal{A}_a$, and $\mathcal{A}_s$ have been sufficiently sampled to curate $H_z^a$ and $H_z^s$ respectively. Then we can obtain $\mathcal{A'}_a$, and $\mathcal{A'}_s$ that are equivalent to $\mathcal{A}_a$,and $\mathcal{A}_s$ respectively with Lemma \ref{lemma:natural_fact} and Equation \eqref{eq:trans_solver}. For a given plant $\mathcal{P}$ and desired prefix-closed regular language $\mathcal{K} \subseteq L(\mathcal{P})$, the FST given by 
\begin{equation}
\label{eq:theorem:resilient_sup-1}
\mathcal{S} = \mathcal{A'}_s^{-1} \circ \mathcal{M}_\mathcal{K}^{-1} \circ \mathcal{A'}_a^{-1}
\end{equation}  
is a resilient supervisor if and only if Equation \eqref{eq:sup_req} holds. A resilient supervisor does not exist. 
\end{corollary}
\begin{proof}
    The proof follows immediately from combining Theorem \ref{theorem:summary} and Theorem 3 from \cite{wang2023attackresilient}. Theorem \ref{theorem:summary} provides equivalent FST models for the unknown attackers and Theorem 3 from \cite{wang2023attackresilient} shows how these attackers can be used to create a resilient supervisor. 
\end{proof}

\vspace{-12pt}
\section{Conclusion}
\label{sec:conclusion}
In conclusion, our study addresses the critical issue of cyber attacks on cyber-physical systems (CPS) by leveraging supervisory control methods enhanced with finite-state transducers (FSTs) and spectral analysis. We developed a novel method capable of modeling and neutralizing broad-spectrum, unknown attackers by analyzing the historical behavior of attacks within a supervisory control loop. Our findings demonstrate the potential of supervisory control strategies to enhance the security and resilience of CPS, by offering a framework to mitigate sophisticated cyber threats and ensure the integrity of critical infrastructures. 
\vspace{-12pt}

\bibliography{nathan_journal}

@incollection{mohri_weighted_2004,
	address = {Berlin, Heidelberg},
	title = {Weighted finite-state transducer algorithms. an overview},
	volume = {148},
	isbn = {978-3-642-53554-3 978-3-540-39886-8},
	language = {english},
	booktitle = {Formal {Languages} and {Applications}},
	publisher = {Springer Berlin Heidelberg},
	author = {Mohri, Mehryar},
	editor = {Kacprzyk, Janusz and Martín-Vide, Carlos and Mitrana, Victor and Păun, Gheorghe},
	year = {2004},
	doi = {10.1007/978-3-540-39886-8_29},
	pages = {551--563},
}

@article{lin_synthesis_2021,
	title = {Synthesis of covert actuator and sensor attackers},
	volume = {130},
	issn = {0005-1098},
	url = {https://www.sciencedirect.com/science/article/pii/S000510982100234X},
	doi = {10.1016/j.automatica.2021.109714},
	language = {en},
	urldate = {2023-04-01},
	journal = {Automatica},
	author = {Lin, Liyong and Su, Rong},
	month = aug,
	year = {2021},
	pages = {109714},
}

@article{wang_state-based_2018,
	title = {State-{Based} {Control} of {Discrete}-{Event} {Systems} {Under} {Partial} {Observation}},
	volume = {6},
	issn = {2169-3536},
	doi = {10.1109/ACCESS.2018.2859798},
	journal = {IEEE Access},
	author = {Wang, Deguang and Lin, Liyong and Li, Zhiwu and Wonham, Walter Murry},
	year = {2018},
	note = {Conference Name: IEEE Access},
	pages = {42084--42093},
}

@article{balle_spectral_2014,
	title = {Spectral learning of weighted automata},
	volume = {96},
	issn = {1573-0565},
	url = {https://doi.org/10.1007/s10994-013-5416-x},
	doi = {10.1007/s10994-013-5416-x},
	language = {en},
	number = {1},
	urldate = {2023-03-31},
	journal = {Machine Learning},
	author = {Balle, Borja and Carreras, Xavier and Luque, Franco M. and Quattoni, Ariadna},
	month = jul,
	year = {2014},
	pages = {33--63},
}

@article{carlyle_realizations_1971,
	title = {Realizations by stochastic finite automata},
	volume = {5},
	issn = {0022-0000},
	url = {https://www.sciencedirect.com/science/article/pii/S0022000071800053},
	doi = {10.1016/S0022-0000(71)80005-3},
	language = {en},
	number = {1},
	urldate = {2023-03-31},
	journal = {Journal of Computer and System Sciences},
	author = {Carlyle, J. W. and Paz, A.},
	month = feb,
	year = {1971},
	pages = {26--40},
}

@inproceedings{balle_learning_2015,
	address = {Cham},
	series = {Lecture {Notes} in {Computer} {Science}},
	title = {Learning {Weighted} {Automata}},
	isbn = {978-3-319-23021-4},
	doi = {10.1007/978-3-319-23021-4_1},
	language = {en},
	booktitle = {Algebraic {Informatics}},
	publisher = {Springer International Publishing},
	author = {Balle, Borja and Mohri, Mehryar},
	editor = {Maletti, Andreas},
	year = {2015},
	pages = {1--21},
}

@book{cassandras2008introduction,
	address = {New York, NY},
	edition = {2. ed},
	title = {Introduction to discrete event systems},
	isbn = {978-0-387-33332-8 978-1-4419-4119-0 978-0-387-68612-7},
	language = {english},
	publisher = {Springer},
	author = {Cassandras, Christos G. and Lafortune, Stéphane},
	year = {2008},
}

@inproceedings{raghupatruni_empirical_2019,
	title = {Empirical {Testing} of {Automotive} {Cyber}-{Physical} {Systems} with {Credible} {Software}-in-the-{Loop} {Environments}},
	url = {https://ieeexplore.ieee.org/abstract/document/8965169},
	doi = {10.1109/ICCVE45908.2019.8965169},
	urldate = {2025-01-16},
	booktitle = {2019 {IEEE} {International} {Conference} on {Connected} {Vehicles} and {Expo} ({ICCVE})},
	author = {Raghupatruni, Indrasen and Goeppel, Thomas and Atak, Muhammed and Bou, Julien and Huber, Thomas},
	month = nov,
	year = {2019},
	note = {ISSN: 2378-1297},
	pages = {1--6},
}

@inproceedings{li_deep_2022,
	title = {Deep {Reinforcement} {Learning} for {Penetration} {Testing} of {Cyber}-{Physical} {Attacks} in the {Smart} {Grid}},
	url = {https://ieeexplore.ieee.org/abstract/document/9892584},
	doi = {10.1109/IJCNN55064.2022.9892584},
	urldate = {2025-01-16},
	booktitle = {2022 {International} {Joint} {Conference} on {Neural} {Networks} ({IJCNN})},
	author = {Li, Yuanliang and Yan, Jun and Naili, Mohamed},
	month = jul,
	year = {2022},
	note = {ISSN: 2161-4407},
	pages = {01--09},
}

@inproceedings{he_system_2019,
	title = {A {System} {Identification} {Based} {Oracle} for {Control}-{CPS} {Software} {Fault} {Localization}},
	url = {https://ieeexplore.ieee.org/abstract/document/8811955},
	doi = {10.1109/ICSE.2019.00029},
	urldate = {2025-01-16},
	booktitle = {2019 {IEEE}/{ACM} 41st {International} {Conference} on {Software} {Engineering} ({ICSE})},
	author = {He, Zhijian and Chen, Yao and Huang, Enyan and Wang, Qixin and Pei, Yu and Yuan, Haidong},
	month = may,
	year = {2019},
	note = {ISSN: 1558-1225},
	pages = {116--127},
}

@article{hsu_spectral_2012,
	series = {{JCSS} {Special} {Issue}: {Cloud} {Computing} 2011},
	title = {A spectral algorithm for learning {Hidden} {Markov} {Models}},
	volume = {78},
	issn = {0022-0000},
	url = {https://www.sciencedirect.com/science/article/pii/S0022000012000244},
	doi = {10.1016/j.jcss.2011.12.025},
	language = {en},
	number = {5},
	urldate = {2023-03-31},
	journal = {Journal of Computer and System Sciences},
	author = {Hsu, Daniel and Kakade, Sham M. and Zhang, Tong},
	month = sep,
	year = {2012},
	pages = {1460--1480},
}

@article{you2022livenessenforcing,
	title = {A liveness-enforcing supervisor tolerant to sensor-reading modification attacks},
	volume = {52},
	issn = {2168-2232},
	doi = {10.1109/TSMC.2021.3051096},
	number = {4},
	urldate = {2023-11-07},
	journal = {IEEE Transactions on Systems, Man, and Cybernetics: Systems},
	author = {You, Dan and Wang, Shouguang and Seatzu, Carla},
	month = apr,
	year = {2022},
	pages = {2398--2411},
}

@inproceedings{yaoSensorDeceptionAttacks2022,
	title = {Sensor {Deception} {Attacks} {Against} {Initial}-{State} {Privacy} in {Supervisory} {Control} {Systems}},
	doi = {10.1109/CDC51059.2022.9992694},
	urldate = {2023-11-07},
	booktitle = {2022 {IEEE} 61st {Conference} on {Decision} and {Control} ({CDC})},
	author = {Yao, Jingshi and Yin, Xiang and Li, Shaoyuan},
	month = dec,
	year = {2022},
	note = {ISSN: 2576-2370},
	pages = {4839--4845},
}

@article{wonhamSupervisoryControlDiscreteevent2018,
	title = {Supervisory control of discrete-event systems: {A} brief history},
	volume = {45},
	issn = {1367-5788},
	shorttitle = {Supervisory control of discrete-event systems},
	doi = {10.1016/j.arcontrol.2018.03.002},
	urldate = {2023-11-14},
	journal = {Annual Reviews in Control},
	author = {Wonham, W. M. and Cai, Kai and Rudie, Karen},
	month = jan,
	year = {2018},
	pages = {250--256},
}

@inproceedings{wang2019supervisory,
	address = {Nice, France},
	title = {Supervisory control of discrete event systems in the presence of sensor and actuator attacks},
	booktitle = {{IEEE} {Conference} on {Decision} and {Control}},
	author = {Wang, Yu and Pajic, Miroslav},
	year = {2019},
	pages = {5350--5355},
}

@inproceedings{whitehead2017ukraine,
	title = {Ukraine cyber-induced power outage: {Analysis} and practical mitigation strategies},
	shorttitle = {Ukraine cyber-induced power outage},
	doi = {10.1109/CPRE.2017.8090056},
	urldate = {2024-01-15},
	booktitle = {2017 70th {Annual} {Conference} for {Protective} {Relay} {Engineers} ({CPRE})},
	author = {Whitehead, David E. and Owens, Kevin and Gammel, Dennis and Smith, Jess},
	month = apr,
	year = {2017},
	note = {ISSN: 2474-9753},
	pages = {1--8},
}

@article{wang2023attackresilient,
	title = {Attack-{Resilient} {Supervisory} {Control} of {Discrete}-{Event} {Systems}: {A} {Finite}-{State} {Transducer} {Approach}},
	volume = {2},
	issn = {2694-085X},
	shorttitle = {Attack-{Resilient} {Supervisory} {Control} of {Discrete}-{Event} {Systems}},
	doi = {10.1109/OJCSYS.2023.3290408},
	urldate = {2023-11-14},
	journal = {IEEE Open Journal of Control Systems},
	author = {Wang, Yu and Bozkurt, Alper Kamil and Smith, Nathan and Pajic, Miroslav},
	year = {2023},
	pages = {208--220},
}

@inproceedings{wang2019attackresilient,
	address = {Nice, France},
	title = {Attack-resilient supervisory control with intermittently secure communication},
	booktitle = {{IEEE} {Conference} on {Decision} and {Control}},
	author = {Wang, Yu and Pajic, Miroslav},
	year = {2019},
	pages = {2015--2020},
}

@article{wakaikiSupervisoryControlDiscreteEvent2019,
	title = {Supervisory {Control} of {Discrete}-{Event} {Systems} {Under} {Attacks}},
	volume = {9},
	issn = {2153-0793},
	doi = {10.1007/s13235-018-0285-3},
	language = {english},
	number = {4},
	urldate = {2023-11-07},
	journal = {Dynamic Games and Applications},
	author = {Wakaiki, Masashi and Tabuada, Paulo and Hespanha, João P.},
	month = dec,
	year = {2019},
	pages = {965--983},
}

@inproceedings{thapliyal2021learning,
	title = {Learning based {Cyberattack} {Design} and {Defense} for {Supervisory} {Control} {Systems}},
	doi = {10.23919/ECC54610.2021.9655174},
	urldate = {2024-01-07},
	booktitle = {2021 {European} {Control} {Conference} ({ECC})},
	author = {Thapliyal, Omanshu and Hwang, Inseok},
	month = jun,
	year = {2021},
	pages = {144--149},
}

@article{taiSynthesisOptimalCovert2023,
	title = {Synthesis of optimal covert sensor–actuator attackers for discrete-event systems},
	volume = {151},
	issn = {0005-1098},
	doi = {10.1016/j.automatica.2023.110910},
	urldate = {2023-11-07},
	journal = {Automatica},
	author = {Tai, Ruochen and Lin, Liyong and Su, Rong},
	month = may,
	year = {2023},
	pages = {110910},
}

@article{taiPrivacypreservingCosynthesisSensor2023,
	title = {Privacy-preserving co-synthesis against sensor–actuator eavesdropping intruder},
	volume = {150},
	issn = {0005-1098},
	doi = {10.1016/j.automatica.2023.110860},
	urldate = {2023-11-14},
	journal = {Automatica},
	author = {Tai, Ruochen and Lin, Liyong and Zhu, Yuting and Su, Rong},
	month = apr,
	year = {2023},
	pages = {110860},
}

@article{tai2023synthesis,
	title = {Synthesis of the {Supremal} {Covert} {Attacker} {Against} {Unknown} {Supervisors} by {Using} {Observations}},
	volume = {68},
	issn = {1558-2523},
	doi = {10.1109/TAC.2022.3191393},
	number = {6},
	urldate = {2024-01-07},
	journal = {IEEE Transactions on Automatic Control},
	author = {Tai, Ruochen and Lin, Liyong and Zhu, Yuting and Su, Rong},
	month = jun,
	year = {2023},
	pages = {3453--3468},
}

@article{su2018supervisor,
	title = {Supervisor synthesis to thwart cyber attack with bounded sensor reading alterations},
	volume = {94},
	issn = {00051098},
	doi = {10.1016/j.automatica.2018.04.006},
	language = {english},
	journal = {Automatica},
	author = {Su, Rong},
	month = aug,
	year = {2018},
	pages = {35--44},
}

@article{soltan2019react,
	title = {{REACT} to {Cyber} {Attacks} on {Power} {Grids}},
	volume = {6},
	issn = {2327-4697},
	doi = {10.1109/TNSE.2018.2837894},
	number = {3},
	urldate = {2024-01-15},
	journal = {IEEE Transactions on Network Science and Engineering},
	author = {Soltan, Saleh and Yannakakis, Mihalis and Zussman, Gil},
	month = jul,
	year = {2019},
	pages = {459--473},
}

@misc{singhStudyCyberAttacks2018,
	title = {Study of {Cyber} {Attacks} on {Cyber}-{Physical} {System}},
	language = {english},
	urldate = {2023-11-14},
	author = {Singh, Ajeet and Jain, Anurag},
	month = apr,
	year = {2018},
	doi = {10.2139/ssrn.3170288},
	note = {Number: 3170288
Place: Rochester, NY
Type: SSRN Scholarly Paper},
}

@article{sharif2022literature,
	title = {A literature review of financial losses statistics for cyber security and future trend},
	volume = {15},
	issn = {2581-9615, 2581-9615},
	doi = {10.30574/wjarr.2022.15.1.0573},
	language = {english},
	number = {1},
	urldate = {2024-01-18},
	journal = {World Journal of Advanced Research and Reviews},
	author = {Sharif, Md Haris Uddin and Mohammed, Mehmood Ali and Sharif, Md Haris Uddin and Mohammed, Mehmood Ali},
	year = {2022},
	note = {Publisher: World Journal of Advanced Research and Reviews
tex.copyright: Copyrights to World Journal of Advanced Research and Reviews},
	pages = {138--156},
}

@inproceedings{rashidinejadSupervisoryControlDiscreteEvent2019,
	title = {Supervisory {Control} of {Discrete}-{Event} {Systems} under {Attacks}: {An} {Overview} and {Outlook}},
	shorttitle = {Supervisory {Control} of {Discrete}-{Event} {Systems} under {Attacks}},
	doi = {10.23919/ECC.2019.8795849},
	urldate = {2023-11-14},
	booktitle = {2019 18th {European} {Control} {Conference} ({ECC})},
	author = {Rashidinejad, Aida and Wetzels, Bart and Reniers, Michel and Lin, Liyong and Zhu, Yuting and Su, Rong},
	month = jun,
	year = {2019},
	pages = {1732--1739},
}

@article{meira-goesSynthesisSensorDeception2020,
	title = {Synthesis of sensor deception attacks at the supervisory layer of {Cyber}–{Physical} {Systems}},
	volume = {121},
	issn = {0005-1098},
	doi = {10.1016/j.automatica.2020.109172},
	urldate = {2023-11-07},
	journal = {Automatica},
	author = {Meira-Góes, Rômulo and Kang, Eunsuk and Kwong, Raymond H. and Lafortune, Stéphane},
	month = nov,
	year = {2020},
	pages = {109172},
}

@article{meira-goes2021synthesis,
	title = {Synthesis of {Supervisors} {Robust} {Against} {Sensor} {Deception} {Attacks}},
	volume = {66},
	issn = {1558-2523},
	doi = {10.1109/TAC.2021.3051459},
	number = {10},
	urldate = {2023-12-23},
	journal = {IEEE Transactions on Automatic Control},
	author = {Meira-Góes, Rômulo and Lafortune, Stéphane and Marchand, Hervé},
	month = oct,
	year = {2021},
	pages = {4990--4997},
}

@article{limaSecurityCyberPhysicalSystems2022,
	title = {Security of {Cyber}-{Physical} {Systems}: {Design} of a {Security} {Supervisor} to {Thwart} {Attacks}},
	volume = {19},
	issn = {1558-3783},
	shorttitle = {Security of {Cyber}-{Physical} {Systems}},
	doi = {10.1109/TASE.2021.3076697},
	number = {3},
	urldate = {2023-11-07},
	journal = {IEEE Transactions on Automation Science and Engineering},
	author = {Lima, Públio M. and Alves, Marcos V. S. and Carvalho, Lilian Kawakami and Moreira, Marcos V.},
	month = jul,
	year = {2022},
	pages = {2030--2041},
}

@article{lin2020synthesis,
	title = {Synthesis of covert actuator attackers for free},
	volume = {30},
	issn = {1573-7594},
	doi = {10.1007/s10626-020-00312-2},
	language = {english},
	number = {4},
	urldate = {2023-12-23},
	journal = {Discrete Event Dynamic Systems},
	author = {Lin, Liyong and Zhu, Yuting and Su, Rong},
	month = dec,
	year = {2020},
	pages = {561--577},
}

@article{langner2011stuxnet,
	title = {Stuxnet: {Dissecting} a {Cyberwarfare} {Weapon}},
	volume = {9},
	issn = {1558-4046},
	shorttitle = {Stuxnet},
	doi = {10.1109/MSP.2011.67},
	number = {3},
	urldate = {2024-01-09},
	journal = {IEEE Security \& Privacy},
	author = {Langner, Ralph},
	month = may,
	year = {2011},
	pages = {49--51},
}

@inproceedings{jazdi2014cyber,
	title = {Cyber physical systems in the context of {Industry} 4.0},
	doi = {10.1109/AQTR.2014.6857843},
	urldate = {2024-01-14},
	booktitle = {2014 {IEEE} {International} {Conference} on {Automation}, {Quality} and {Testing}, {Robotics}},
	author = {Jazdi, N.},
	month = may,
	year = {2014},
	pages = {1--4},
}

@book{holcombe2004algebraic,
	title = {Algebraic automata theory},
	isbn = {978-0-521-60492-5},
	language = {english},
	publisher = {Cambridge University Press},
	author = {Holcombe, M.},
	month = jun,
	year = {2004},
	note = {tex.googlebooks: xCtthlOQUfYC},
}

@article{duo2022survey,
	title = {A {Survey} of {Cyber} {Attacks} on {Cyber} {Physical} {Systems}: {Recent} {Advances} and {Challenges}},
	volume = {9},
	issn = {2329-9274},
	shorttitle = {A {Survey} of {Cyber} {Attacks} on {Cyber} {Physical} {Systems}},
	doi = {10.1109/JAS.2022.105548},
	number = {5},
	urldate = {2024-01-10},
	journal = {IEEE/CAA Journal of Automatica Sinica},
	author = {Duo, Wenli and Zhou, MengChu and Abusorrah, Abdullah},
	month = may,
	year = {2022},
	pages = {784--800},
}

@article{diazredondo2020security,
	title = {Security {Aspects} in {Smart} {Meters}: {Analysis} and {Prevention}},
	volume = {20},
	issn = {1424-8220},
	shorttitle = {Security {Aspects} in {Smart} {Meters}},
	doi = {10.3390/s20143977},
	language = {english},
	number = {14},
	urldate = {2024-01-09},
	journal = {Sensors},
	author = {Díaz Redondo, Rebeca P. and Fernández-Vilas, Ana and Fernández dos Reis, Gabriel},
	month = jan,
	year = {2020},
	note = {Publisher: Multidisciplinary Digital Publishing Institute},
	pages = {3977},
}

@article{carvalho2018detection,
	title = {Detection and mitigation of classes of attacks in supervisory control systems},
	volume = {97},
	issn = {0005-1098},
	doi = {10.1016/j.automatica.2018.07.017},
	journal = {Automatica},
	author = {Carvalho, Lilian Kawakami and Wu, Yi-Chin and Kwong, Raymond and Lafortune, Stéphane},
	month = nov,
	year = {2018},
	pages = {121--133},
}

@inproceedings{battistelli2009unfalsified,
	title = {Unfalsified adaptive switching supervisory control of time varying systems},
	doi = {10.1109/CDC.2009.5399882},
	urldate = {2024-01-05},
	booktitle = {Proceedings of the 48h {IEEE} {Conference} on {Decision} and {Control} ({CDC}) held jointly with 2009 28th {Chinese} {Control} {Conference}},
	author = {Battistelli, Giorgio and Hespanha, João and Mosca, Edoardo and Tesi, Pietro},
	month = dec,
	year = {2009},
	note = {ISSN: 0191-2216},
	pages = {805--810},
}

@article{angluinLearningRegularSets1987,
	title = {Learning regular sets from queries and counterexamples},
	volume = {75},
	issn = {0890-5401},
	doi = {10.1016/0890-5401(87)90052-6},
	language = {english},
	number = {2},
	journal = {Information and Computation},
	author = {Angluin, Dana},
	month = nov,
	year = {1987},
	note = {tex.ids: angluin\_LearningRegularSets\_1987},
	pages = {87--106},
}

@article{alguliyevCyberphysicalSystemsTheir2018,
	title = {Cyber-physical systems and their security issues},
	volume = {100},
	issn = {0166-3615},
	doi = {10.1016/j.compind.2018.04.017},
	urldate = {2023-11-14},
	journal = {Computers in Industry},
	author = {Alguliyev, Rasim and Imamverdiyev, Yadigar and Sukhostat, Lyudmila},
	month = sep,
	year = {2018},
	pages = {212--223},
}

@article{zhang_perception_2023,
	title = {Perception and sensing for autonomous vehicles under adverse weather conditions: {A} survey},
	volume = {196},
	issn = {0924-2716},
	shorttitle = {Perception and sensing for autonomous vehicles under adverse weather conditions},
	url = {https://www.sciencedirect.com/science/article/pii/S0924271622003367},
	doi = {10.1016/j.isprsjprs.2022.12.021},
	urldate = {2024-08-20},
	journal = {ISPRS Journal of Photogrammetry and Remote Sensing},
	author = {Zhang, Yuxiao and Carballo, Alexander and Yang, Hanting and Takeda, Kazuya},
	month = feb,
	year = {2023},
	pages = {146--177},
}
\bibliographystyle{IEEEtran}



\end{document}